\documentclass[pre,twocolumn,10pt,aps,longbibliography,nofootinbib]{revtex4-1}

 \usepackage{verbatim}
 \usepackage{amsmath}
 \usepackage{amssymb}
 \usepackage{amsthm}
 \usepackage{latexsym}
 \usepackage{amsfonts}
 \usepackage{epsfig}
 \usepackage{epstopdf}
 \usepackage{wasysym} 
 \usepackage{color}
 \definecolor{darkblue}{rgb}{0,0,.5}
 \usepackage[linktocpage, colorlinks=true, linkcolor=darkblue, citecolor=darkblue]{hyperref}
 \usepackage[all]{hypcap}
 \usepackage[makeroom]{cancel}

\newcommand{\C}[1]{{\cal{#1}}}
\newcommand{\bb}[1]{\textbf{#1}}

\newcommand{\mf}[1]{{\mathfrak{#1}}}

\newcommand{\lr}[1]{{\left\langle {#1}\right\rangle}}
\newcommand{\rl}[0]{{\rangle\langle}}

\begin{document}

\title{Classical quantum stochastic processes}

\author{Philipp Strasberg}
\author{Mar\'ia Garc\'ia D\'iaz}
\affiliation{F\'isica Te\`orica: Informaci\'o i Fen\`omens Qu\`antics, Departament de F\'isica, Universitat Aut\`onoma de Barcelona, 08193 Bellaterra (Barcelona), Spain}

\date{\today}

\begin{abstract}
 We investigate the role of coherence and Markovianity in finding an answer to the question whether the outcomes of 
 a projectively measured quantum stochastic process are compatible with a classical stochastic process. 
 For this purpose we put forward an operationally motivated definition of incoherent dynamics applicable to any open 
 system's dynamics. For non-degenerate observables described by rank-1 projective measurements we show that 
 classicality always implies incoherent dynamics, whereas the converse is only true for invertible Markovian (but not 
 necessarily time-homogeneous) dynamics. For degenerate observables the picture is somewhat reversed as classicality 
 does no longer suffice to imply incoherent dynamics (even in the invertible Markovian case), while an incoherent, 
 invertible Markovian dynamics still implies classicality. 
\end{abstract}

\maketitle

\newtheorem{mydef}{Definition}[section]
\newtheorem{lemma}{Lemma}[section]
\newtheorem{thm}{Theorem}[section]
\newtheorem{crllr}{Corollary}[section]
\newtheorem*{thm*}{Theorem}
\theoremstyle{remark}
\newtheorem{rmrk}{Remark}[section]

\section{Introduction}

Although in actual experiments with classical systems it might not always be possible to measure the system without 
disturbing it, at least theoretically one can consider the ideal limit of a non-invasive measurement. 
This idea has led to the theory of stochastic processes, a major mathematical toolbox used across many scientific 
disciplines~\cite{VanKampenBook2007, GardinerBook2009}. Since the limit of an ideal non-disturbing measurement 
does not exist for quantum systems, a widely accepted consensus of what a quantum stochastic process 
actually \emph{is} has not yet emerged. However, recent progress (see Ref.~\cite{MilzEtAlArXiv2017} and references 
therein) strongly suggests that a quantum stochastic process is conceptually similar to classical causal 
modeling~\cite{PearlBook2009} and here we will follow this approach. Understanding under which circumstances a 
projectively measured quantum system can be effectively described in a classical way is therefore of fundamental 
interest as it sheds light on the gap between quantum and classical stochastic processes. In addition, it enables 
us to distinguish quantum from classical features which is a relevant task for future technologies (e.g., in quantum 
information or quantum thermoydnamics) and for the field of quantum biology. Finally, it also has practical relevance 
as classical stochastic processes are easier to simulate. 

The relation between classical and quantum stochastic processes was first addressed by Smirne and 
co-workers~\cite{SmirneEtAlQST2018}, who showed that the 
answer to the question whether a quantum system effectively behaves classical is closely related to the question 
whether coherences play a role in its evolution. More specifically, for a quantum dynamical semigroup obeying 
the regression theorem (i.e., a time-homogeneous quantum Markov process), it was shown that the statistics obtained 
from rank-1 projective measurements of a given system observable are compatible with a classical stochastic process 
if and only if the dynamics is ``non-coherence-generating-and-detecting (NCGD)''~\cite{SmirneEtAlQST2018}. 

The purpose of the present paper is to extend the results of Smirne \emph{et al.}~in various directions. 
We will provide an operationally motivated definition of \emph{incoherent dynamics}, which is supposed to capture 
the absence of any detectable coherence in the dynamics. It is applicable to any open systems dynamics and it is 
different from the NCGD notion. Our definition allows us to prove the following: 
first, for non-degenerate observables described by rank-1 projectors, any process which can be effectively described 
classically is incoherent (i.e., cannot generate any detectable coherence), whereas the converse is only true for 
invertible Markovian, but not necessarily time-homogeneous dynamics. Second, for degenerate observables, we lose the 
property that classicality implies incoherent dynamics because detectable coherence can be hidden in degenerate 
subspaces. 

The rest of the paper is structured as follows. In Sec.~\ref{sec mathematical preliminaries} we set the stage and 
introduce some basic definitions. Our main results are reported in Sec.~\ref{sec non degenerate} for non-degenerate 
observables and in Sec.~\ref{sec degenerate} for degenerate observables. We conclude in Sec.~\ref{sec conclusions}. 
A thorough comparison with the framework of Ref.~\cite{SmirneEtAlQST2018} is given in Appendix~\ref{sec app comparison} 
showing that our results reduce to the ones of Smirne \emph{et al.}~in the respective limit. Various counterexamples, 
which demonstrate that our main theorems in Sec.~\ref{sec results} are tight, are postponed to 
Appendix~\ref{sec app counterexamples}.

\section{Mathematical preliminaries}
\label{sec mathematical preliminaries}

We start by reviewing basic notions of a classical stochastic process. We label the classical distinguishable 
states of the system of interest by $r$ and we assume that the system gets measured at an arbitrary set of discrete 
times $\{t_1,\dots,t_n\}$. We denote the result at time $t_i$ by $r_i$. Furthermore, for reasons which will become 
clearer later on, we explicitly denote the initial preparation of the experiment by $\C A_0$. At 
this stage the reader can think of this as merely a verbal description of how to initialize the experiment 
(e.g., `wait long enough such that the system is equilibrated and start measuring afterwards'), later on it will 
mathematically turn out to be a completely positive and trace-preserving map. We then denote the 
joint probability distribution to get the sequence of measurement results $\bb r_n = r_1,\dots,r_n$ at times 
$t_1,\dots,t_n$ given the initial preparation $\C A_0$ by 
\begin{equation}\label{eq joint probability}
 p(r_n,t_n;\dots;r_1,t_1|\C A_0) \equiv p(\bb r_n|\C A_0).
\end{equation}
The following definition is standard: 

\begin{mydef}
 The probabilities $p(\bb r_n|\C A_0)$ are said to be classical with repect to a given preparation procedure 
 $\C A_0$ if they fulfill the consistency condition 
 \begin{equation}\label{eq def classicality}
  \sum_{r_k} p(r_\ell,\dots,r_k,\dots,r_j|\C A_0) = p(r_\ell,\dots,\cancel{r_k},\dots,r_j|\C A_0)
 \end{equation}
 for all $\ell\ge k\ge j\ge1$. Here, the probability on the right hand side is constructed by measuring the states 
 $r_i$ of the system only at the set of times $\{t_\ell,\dots,t_j\}\setminus\{t_k\}$. 
\end{mydef}

We remark that, if the consistency requirement~(\ref{eq def classicality}) is fulfilled, then -- by the 
Kolmogorov-Daniell extension theorem -- we know that there exists an underlying continuous-in-time 
stochastic process, which contains all joint proabilities~(\ref{eq joint probability}) as marginals. The importance 
of this theorem lies in the fact that it allows us to bridge experimental reality (where any measurement statistics 
is always finite) with its theoretical description (which often uses continuous-time dynamics in form of, e.g., 
stochastic differential equations). 

Albeit condition~(\ref{eq def classicality}) is in general not fulfilled for quantum dynamics, the joint probability 
distribution~(\ref{eq joint probability}) is nevertheless a well-defined object in quantum mechanics. For this purpose 
we assume that the experimentalist measures at time $t_k$ an arbitrary system observable $R_k = \sum_{r_k} r_k P_{r_k}$ 
with projectors $P_{r_k} = P_{r_k}^2$ and eigenvalues $r_k\in\mathbb{R}$. If all projectors are rank-1, i.e., 
$P_{r_k} = |r_k\rl r_k|$, we talk about 
a non-degenerate system observable, otherwise we call it degenerate. Furthermore, following the conventional picture of 
open quantum systems~\cite{BreuerPetruccioneBook2002}, we allow the system $S$ to be coupled to an arbitrary environment 
$E$. The initial system-environment state at time $t_0<t_1$ is denoted by $\rho_{SE}(t_0)$. Then, by using superoperator 
notation, we can express Eq.~(\ref{eq joint probability}) as 
\begin{equation}\label{eq probability}
 \begin{split}
  & p(\bb r_n|\C A_0)    \\
  &=  \mbox{tr}_{SE}\left\{\C P_{r_n}\C U_{n,n-1}\dots \C P_{r_2}\C U_{2,1}\C P_{r_1}\C U_{1,0}\C A_0\rho_{SE}(t_0)\right\} \\
  &\equiv \mbox{tr}_S\left\{\mf T_{n+1}[\C P_{r_n},\dots,\C P_{r_2},\C P_{r_1},\C A_0]\right\}.
 \end{split}
\end{equation}
Here, the preparation procedure $\C A_0$ is an arbitrary completely positive (CP) and trace-preserving map acting on the 
system only (we suppress identity operations in the tensor product notation). Notice that the preparation procedure 
could itself be the identity operation (i.e., `do nothing') denoted by $\C A_0 = \C I_0$. Furthermore, $\C U_{k,k-1}$ 
denotes the unitary time-evolution propagating the system-environment state from time $t_{k-1}$ to $t_k$ (we make no 
assumption about the underlying Hamiltonian here). We also introduced the projection superoperator 
$\C P_{r_k} \rho \equiv P_{r_k}\rho P_{r_k}$, which acts only on the system and corresponds to result $r_k$ at time 
$t_k$. Finally, in the last line of Eq.~(\ref{eq probability}) we have introduced the $(n+1)$-step `process tensor' 
$\mf T_{n+1}$~\cite{PollockEtAlPRA2018} (also called `quantum comb'~\cite{ChiribellaDArianoPerinottiPRL2008, 
ChiribellaDArianoPerinottiPRA2009} or `process matrix'~\cite{CostaShrapnelNJP2016, OreshkovGiarmatziNJP2016}). It is 
a formal but operationally well-defined object: it yields the (subnormalized) state of the system 
$\tilde\rho_S(\C P_{r_n},\dots,\C P_{r_2},\C P_{r_1},\C A_0) = \mf T_{n+1}[\C P_{r_n},\dots,\C P_{r_2},\C P_{r_1},\C A_0]$ 
conditioned on a certain sequence of interventions $\C P_{r_n},\dots,\C P_{r_2},\C P_{r_1},\C A_0$. Its norm, 
as given by the trace over $S$, equals the probability to obtain the measurement results $\bb r_n$. Recently, 
it was shown that the process tensor allows for a rigorous definition of quantum stochastic processes (or quantum causal 
models) fulfilling a generalized version of the Kolmogorov-Daniell extension theorem~\cite{MilzEtAlArXiv2017}. 
We also add that complete knowledge of the process tensor $\mf T_n$ implies complete knowledge of the process 
tensor $\mf T_\ell$ for $\ell\le n$, i.e., $\mf T_n$ contains $\mf T_\ell$. 

We now have the main tools at hand to precisely state the question we are posing in this paper: Which conditions 
does a quantum stochastic process need to fulfill in order to guarantee that the resulting measurement statistics can 
(or cannot) be explained by a classical stochastic process? That is, when is Eq.~(\ref{eq def classicality}) fulfilled 
or, in terms of the process tensor, when is 
\begin{equation}
 \begin{split}\label{eq classicality process tensor}
  & \mbox{tr}_S\{\mf T_{\ell+1}[\C P_{r_\ell},\dots,\Delta_k,\dots,\C P_{r_j},\dots,\C A_0]\} \\
  &\stackrel{?}{=} \mbox{tr}_S\{\mf T_{\ell+1}[\C P_{r_\ell},\dots,\C I_k,\dots,\C P_{r_j},\dots,\C A_0]\}.
 \end{split}
\end{equation}
Here, we have introduced the \emph{dephasing operation} at time $t_k$, $\Delta_k \equiv \sum_{r_k} \C P_{r_k}$,
which plays an essential role in the following. Furthermore, the dots in Eq.~(\ref{eq classicality process tensor}) 
denote either projective measurements (if the system gets measured at that time) or identity operations (if the 
system does not get measured at that time). 

To answer the question, we will need a suitable notion of an `incoherent' quantum stochastic process, defined as 
follows: 

\begin{mydef}
 For a given set of observables $\{R_k\}$, $k\in\{1,\dots,\ell\}$, we call the dynamics of an open quantum system 
 $\ell$-incoherent with respect to the preparation $\C A_0$ if all process tensors 
 \begin{equation}\label{eq incoherent}
  \mf T_{\ell+1}\left[\Delta_\ell,\left\{\begin{matrix} \Delta_{\ell-1} \\ \C I_{\ell-1} \\ \end{matrix}\right\},\dots, 
  \left\{\begin{matrix} \Delta_1 \\ \C I_1 \\ \end{matrix}\right\},\C A_0\right]
 \end{equation}
 are equal. Here, the angular bracket notation means that at each time step we can freely choose to perform either 
 a dephasing operation ($\Delta$) or nothing ($\C I$). If the dynamics are $\ell$-incoherent for all 
 $\ell\in\{1,\dots,n\}$, we simply call the dynamics incoherent with respect to the preparation procedure $\C A_0$. 
\end{mydef}

This definition is supposed to capture the situation where the experimentalist has no ability to detect the presence 
of coherence during the course of the evolution. For this purpose we imagine that the experimentalist can manipulate 
the system in two ways: first, she can prepare the initial system state in some way via $\C A_0$ (which could be 
only the identity operation) and she can projectively measure the system observables $R_k$ at times 
$t_k\in\{t_1,\dots,t_n\}$. The question is then: if the final state got dephased with respect to the observable $R_\ell$ 
(e.g., by performing a final measurement of $R_\ell$), is the experimentalist able to infer whether the system was 
subjected to additional dephasing operations at earlier times, i.e., can possible coherences at earlier times become 
manifest in different populations at the final time $t_\ell$? If that is not the case, the dynamics are called 
$\ell$-incoherent. We remark that a process that is $\ell$-incoherent is not necessarily $k$-incoherent for $k\neq\ell$. 
It is therefore important to specify at which (sub)set of times the process is incoherent. In the following we will be 
only interested in processes which are $\ell$-incoherent for all $\ell\in\{1,\dots,n\}$, henceforth dubbed simply 
`incoherent' (with respect to the preparation $\C A_0$). We repeat that our definition of incoherence is different 
from the NCGD notion introduced in Ref.~\cite{SmirneEtAlQST2018}, see Appendix~\ref{sec app comparison}. 
Furthermore, a similar idea restricted to two times was introduced in Ref.~\cite{GessnerBreuerPRL2011} in order 
to detect nonclassical system-environment correlations in the dynamics of open quantum systems. 

\section{Results}
\label{sec results}

\subsection{Non-degenerate observables}
\label{sec non degenerate}

Our first main result is the following: 

\begin{thm}\label{thm classical implies incoherent}
 If the measurement statistics are classical with respect to $\C A_0$, then the dynamics is incoherent with respect 
 to $\C A_0$. 
\end{thm}

Before we prove it, we remark that this theorem holds for any quantum stochastic process (especially without imposing 
Markovianity). Furthermore, a classical process for the times $\{t_n,\dots,t_1\}$ is also classical for all subsets of 
times and hence, the theorem implies incoherence, i.e., $\ell$-incoherence for all $\ell\in\{1,\dots,n\}$. 
In the following proof we will only display the case $\ell=n$, as the rest follows immediately. 

\begin{proof}
 We start by noting that 
 \begin{equation}\label{eq step 1 theorem 1}
  \mf T_{n+1}[\C P_{r_n},\dots,\C P_{r_1},\C A_0] = p(r_n,\dots,r_1|\C A_0) |r_n\rl r_n|,
 \end{equation}
 which is a general identity as we have not made any assumption about the joint probability $p(r_n,\dots,r_1|\C A_0)$. 
 Obviously, if we choose to perform nothing at any time $t_\ell<t_n$, we have 
 \begin{equation}
  \begin{split}
   & \mf T_{n+1}[\C P_{r_n},\dots,\C I_\ell,\dots,\C P_{r_1},\C A_0] \\
   &= p(r_n,\dots,\cancel{r_\ell},\dots,r_1|\C A_0)|r_n\rl r_n|.
  \end{split}
 \end{equation}
 But by assumption of classicality, this is equal to 
 \begin{equation}\label{eq help 1}
  \begin{split}
   & \mf T_{n+1}[\C P_{r_n},\dots,\C I_\ell,\dots,\C P_{r_1},\C A_0] \\
   &= \sum_{r_\ell} p(r_n,\dots,r_\ell,\dots,r_1|\C A_0)|r_n\rl r_n| \\
   &= \sum_{r_\ell} \mf T_{n+1}[\C P_{r_n},\dots,\C P_{r_\ell},\dots,\C P_{r_1},\C A_0]  \\
   &= \mf T_{n+1}[\C P_{r_n},\dots,\Delta_\ell,\dots,\C P_{r_1},\C A_0].
  \end{split}
 \end{equation}
 Hence, by summing Eq.~(\ref{eq help 1}) over the remaining $r_k\neq r_\ell$, we confirm 
 \begin{equation}
  \begin{split}
   & \mf T_{n+1}[\Delta_n,\dots,\C I_\ell,\dots,\Delta_1,\C A_0] \\
   &= \mf T_{n+1}[\Delta_n,\dots,\Delta_\ell,\dots,\Delta_1,\C A_0]
  \end{split}
 \end{equation}
 for arbitrary $t_\ell < t_n$ and where the dots denote dephasing operations at the remaining times. We can now pick 
 another arbitrary time $t_k\neq t_\ell$ and repeat essentially the same steps as above to arrive at the conclusion 
 \begin{equation}
  \begin{split}
   & \mf T_{n+1}[\Delta_n,\dots,\C I_\ell,\dots,\C I_k,\dots,\Delta_1,\C A_0] \\
   &= \mf T_{n+1}[\Delta_n,\dots,\Delta_\ell,\dots,\Delta_k,\dots,\Delta_1,\C A_0]
  \end{split}
 \end{equation}
 for any two times $t_k\neq t_\ell$. By repeating this argument further, we finally confirm that the dynamics are 
 incoherent. 
\end{proof}

The converse of Theorem~\ref{thm classical implies incoherent} holds only in a stricter sense. For this purpose we 
need the notion of Markovianity as defined in Ref.~\cite{PollockEtAlPRL2018}. In there, it was shown that the 
definition of a quantum Markov process implies the notion of \emph{operational CP divisibility}. This means that 
for an arbitrary set of independent interventions (CP maps) 
$\C A_{r_n},\dots,\C A_{r_0}$ the process tensor `factorizes' as 
\begin{equation}\label{eq Markov process tensor}
 \mf T_{n+1}[\C A_{r_n},\dots,\C A_{r_0}] = \C A_{r_n}\Lambda_{n,n-1}\dots\Lambda_{1,0}\C A_{r_0}\rho_S(t_0).
\end{equation}
Here, the set $\{\Lambda_{\ell,k}\}$ is a family of CP and trace-preserving maps fulfilling the composition law 
$\Lambda_{\ell,j} = \Lambda_{\ell,k}\Lambda_{k,j}$ for any $\ell>k>j$. We remark that a CP divisible process (which is 
commonly refered to as being `Markovian') is in general not operationally CP divisible (also see the recent 
discussion in Ref.~\cite{MilzEtAlArXiv2019}). In a nutshell, an operationally CP divisible process always fulfills the 
quantum regression theorem, but a CP divisible process does not (a counterexample is in fact shown in 
Appendix~\ref{sec app comparison}). 

Furthermore, to establish the converse of Theorem~\ref{thm classical implies incoherent} we also need the following 
definition: 

\begin{mydef}
 A Markov process $\{\Lambda_{\ell,k}\}$ is said to be invertible, if the inverse of any $\Lambda_{k,0}$ exists 
 for all $k$, i.e., the CP and trace-preserving maps $\Lambda_{\ell,k}$ are identical to 
 $\Lambda_{\ell,0}\Lambda_{k,0}^{-1}$.
\end{mydef}

We are now ready to prove the next main theorem: 

\begin{thm}\label{thm incoherent classical}
 If the dynamics are Markovian, invertible and incoherent for all preparations $\C A_0$, then the statistics are 
 classical for any preparation. 
\end{thm}

\begin{proof}
 By using Eq.~(\ref{eq Markov process tensor}) and the property of incoherence, we can conclude that for any two 
 times $t_{\ell+1}, t_\ell \in\{t_1,\dots,t_n\}$ (with $t_{\ell+1} > t_\ell$) 
 \begin{equation}
  \Delta_{\ell+1} \Lambda_{\ell+1,\ell}\Delta_\ell \Lambda_{\ell,0}\C A_0\rho_S(t_0) 
  = \Delta_{\ell+1} \Lambda_{\ell+1,\ell} \Lambda_{\ell,0}\C A_0\rho_S(t_0).
 \end{equation}
 Since the dynamics are invertible and incoherent for all preparations $\C A_0$, this implies the superoperator identity 
 $\Delta_{\ell+1} \Lambda_{\ell+1,\ell}\Delta_\ell = \Delta_{\ell+1} \Lambda_{\ell+1,\ell}$. By multiplying this 
 equation with $\C P_{r_{\ell+1}}$, we arrive at 
 \begin{equation}\label{eq superoperator identity}
  \sum_{r_\ell} \C P_{r_{\ell+1}}\Lambda_{\ell+1,\ell}\C P_{r_\ell} = \C P_{r_{\ell+1}}\Lambda_{\ell+1,\ell}. 
 \end{equation}
 From this general identity we immediately obtain that 
 \begin{equation}
  \begin{split}
   & \sum_{r_\ell}p(\bb r_n) \\
   &= \mbox{tr}\left\{\C P_{r_n}\Lambda_{n,n-1}\dots\sum_{r_\ell}\C P_{r_{\ell+1}}\Lambda_{\ell+1,\ell}\C P_{r_\ell}
   \dots\C P_{r_1}\Lambda_{1,0}\C A_0\rho\right\}    \\
   &=   \mbox{tr}\{\C P_{r_n}\Lambda_{n,n-1}\dots\C P_{r_{\ell+1}}\Lambda_{\ell+1,\ell}
   \dots\C P_{r_1}\Lambda_{1,0}\C A_0\rho\}    \\
   &= p(r_n,\dots,\cancel{r_\ell},\dots,r_1).
  \end{split}
 \end{equation}
 This concludes the proof as the above argument also holds for all possible subsets of times. 
\end{proof}

We add that the counterexamples in Appendix~\ref{sec app counterexamples} demonstrate that 
Theorem~\ref{thm incoherent classical} is also tight in the sense that a process, which is incoherent only for a subset 
of preparations or which is not invertible, does not imply classical statistics. One remaining open 
question concerns the assumption of Markovianity. At the moment it is not clear whether relaxing this condition is 
meaningful as it requires to define the notion of invertibility for a non-Markovian process, which is not unambiguous. 

Furthermore, the superoperator identity~(\ref{eq superoperator identity}) implies that, if we write $\Lambda_{\ell,k}$ 
as a matrix in an ordered basis where populations precede coherences with respect to the measured observable $R_k$ 
(input) and $R_\ell$ (output), it has the form 
\begin{equation}
 \Lambda_{\ell,k} = \begin{pmatrix}
                    A_{\ell,k}  &   0   \\
                    C_{\ell,k}  &   D_{\ell,k}   \\
                 \end{pmatrix},
\end{equation}
where $A_{\ell,k}$ is a stochastic matrix and $C_{\ell,k}$ and $D_{\ell,k}$ are matrices, which are only 
constrained by the requirement of complete positivity. 

\subsection{Degenerate observables}
\label{sec degenerate}

If the measured observable contains degeneracies, the picture above somewhat reverses. First, 
Theorem~\ref{thm classical implies incoherent} ceases to hold even in the Markovian and invertible regime because 
the assumption of a non-degenerate observable already entered in the first step of its proof, see 
Eq.~(\ref{eq step 1 theorem 1}). Physically speaking, the reason is that it now becomes possible to hide coherences in 
degenerate subspaces and this can have a detectable effect on the output state~(\ref{eq incoherent}). This is 
demonstrated with the help of an example in Appendix~\ref{sec app counterexamples}. In contrast, 
Theorem~\ref{thm incoherent classical} still holds true for degenerate observables. In fact, in the proof of 
Theorem~\ref{thm incoherent classical} we never used that the measured observable is non-degenerate. 

\section{Conclusions}
\label{sec conclusions}

We have investigated whether the outcomes of a projectively measured quantum system can be described classically 
depending on the capability of an open quantum system to show detectable effects of coherence. The question whether 
the quantum stochastic process is (invertible) Markovian and whether the measured observables are degenerate (or not) 
had a crucial influence on the results. Together with the counterexamples in Appendix~\ref{sec app counterexamples} 
we believe that we have provided a fairly complete picture about the interplay between classicality, coherence 
and Markovianity. It remains, however, still open whether our definition of `incoherent dynamics' is the most 
meaningful one. One clear advantage of our proposal is that it is operationally and theoretically well-defined for 
arbitrary quantum processes. Therefore, it could help to extend existing resource theories, which crucially 
rely on the existence of dynamical maps~\cite{StreltsovAdessoPlenioRMP2017}, to arbitrary multi-time processes.

We further point out that our investigation is closely related to the study of Leggett-Garg inequalities and possible 
violations thereof~\cite{LeggettGargPRL1985, EmaryLambertNoriRPP2014}. In fact, the classicality 
assumption~(\ref{eq def classicality}) plays a crucial role in deriving any Leggett-Garg inequality. 
Therefore, we can conclude that all incoherent quantum systems, which evolve in an invertible Markovian way, will 
never violate a Leggett-Garg inequality if the measured observable is non-degenerate. Interestingly, incoherent 
quantum systems could potentially violate Leggett-Garg inequalities if the measured observable is degenerate. 

Another interesting open point of investigation concerns the question whether the property of incoherence implies 
a particular structure on the generator of a quantum master equation, which is still the primarily used tool in open 
quantum system theory. This question is indeed further pursued by one of us~\cite{Diaz}. 

\emph{Note added.} While this manuscript was under review, we became aware of the work of 
Milz \emph{et al.}~\cite{MilzEgloffEtAlArXiv2019} where an identical question is analysed from a related perspective.

\subsection*{Acknowledgments}

PS is financially supported by the DFG (project STR 1505/2-1) and MGD by `la Caixa' Foundation, grant LCF/BQ/DE16/11570017. We also acknowledge funding from the Spanish MINECO FIS2016-80681-P (AEI-FEDER, UE). 


\bibliography{/home/philipp/Documents/references/books,/home/philipp/Documents/references/open_systems,/home/philipp/Documents/references/thermo,/home/philipp/Documents/references/info_thermo,/home/philipp/Documents/references/general_QM,/home/philipp/Documents/references/math_phys}

\appendix
\section{Comparison with the framework of Smirne \emph{et al.}}
\label{sec app comparison}

In Ref.~\cite{SmirneEtAlQST2018} the notion of ``non-coherence-generating-and-detecting dynamics'' (NCGD dynamics) 
was introduced based on the following definition: 

\begin{mydef}
 The dynamics of an open quantum system is called NCGD with respect to the set of observables $\{R_k\}$ if 
 \begin{equation}\label{eq NCGD}
  \Delta_\ell\Lambda_{\ell,k}\Delta_k\Lambda_{k,j}\Delta_j = \Delta_\ell\Lambda_{\ell,j}\Delta_j
 \end{equation}
 for all $t_\ell\ge t_k\ge t_j\ge t_1$. 
\end{mydef}

In this definition $\Lambda_{\ell,k}$ denotes the `dynamical map' of the quantum system from time $t_k$ to 
time $t_\ell$. For instance, for a time-dependent master equation with Liouvillian $\C L(t)$ this is defined as 
\begin{equation}
 \Lambda_{\ell,k} = \C T_+ \exp\left[\int_{t_k}^{t_\ell} \C L(t) dt\right],
\end{equation}
where $\C T_+$ denotes the time-ordering operator. 

To compare the notions of NCGD and incoherent dynamics, we start by noting that both are almost identical if the 
dynamics are Markovian, invertible and subjected to measurements of a non-degenerate system observable. This is 
important as we are thereby able to confirm the results of Ref.~\cite{SmirneEtAlQST2018} in an independent way. 
To see this, we first prove the following statement: 

\begin{thm}\label{thm incoherent NCGD}
 If the dynamics are Markovian, invertible and incoherent for all possible preparations, then they are also NCGD. 
\end{thm}

\begin{proof}
 By assumption of incoherence we have for an arbitrary preparation $\C A_0$ and an arbitrary set of times 
 $\{t_\ell,t_k,t_j\}$ with $\ell\ge k\ge j\ge1$ 
 \begin{equation}
  \begin{split}
   & \mf T_{\ell+1}[\Delta_\ell,\dots,\Delta_k,\dots,\Delta_j,\dots,\C A_0] \\
   & = \mf T_{\ell+1}[\Delta_\ell,\dots,\C I_k,\dots,\Delta_j,\dots,\C A_0],
  \end{split}
 \end{equation}
 where the dots denote identity operations. By Markovianity, this means that 
 \begin{equation}
  \Delta_\ell\Lambda_{\ell,k}\Delta_k\Lambda_{k,j}\Delta_j\Lambda_{j,0}\C A_0\rho_0
  = \Delta_\ell\Lambda_{\ell,j}\Delta_j\Lambda_{j,0}\C A_0\rho_0.
 \end{equation}
 Since $\C A_0$ is arbitrary and the dynamics are assumed to be invertible, this implies 
 \begin{equation}
  \Delta_\ell\Lambda_{\ell,k}\Delta_k\Lambda_{k,j}\Delta_j = \Delta_\ell\Lambda_{\ell,j}\Delta_j.
 \end{equation}
 Hence, the dynamics are NCGD. 
\end{proof}

The `converse' of Theorem~\ref{thm incoherent NCGD} reads as follows 

\begin{thm}
 If the dynamics is Markovian and NCGD, the dynamics is incoherent with respect to all preparations that result 
 in a diagonal state (with respect to the observable $R_1$) at time $t_1$. 
\end{thm}

\begin{proof}
 Since the dynamics is Markovian and the state at time $t_1$ is diagonal, we always have 
 \begin{equation}
  \begin{split}
   & \mf T_{n+1}\left[\Delta_n,\left\{\begin{matrix} \Delta_{n-1} \\ \C I_{n-1} \\ \end{matrix}\right\},\dots, 
   \left\{\begin{matrix} \Delta_1 \\ \C I_1 \\ \end{matrix}\right\},\C A_0\right] \\
   &= \mf T_{n+1}\left[\Delta_n,\left\{\begin{matrix} \Delta_{n-1} \\ \C I_{n-1} \\ \end{matrix}\right\},\dots,\Delta_1,\C A_0\right].
  \end{split}
 \end{equation}
 Hence, the dynamics are `sandwiched' by two dephasing operations at the beginning at time $t_1$ and at the end at 
 time $t_n$. By the property of NCGD, we are allowed to introduce arbitrary dephasing/identity operations at 
 any time step $t_k$, $n>k>1$. Hence, the dynamics are incoherent. 
\end{proof}

This proves that our main results are not in contradiction to Ref.~\cite{SmirneEtAlQST2018}: In there it was shown that 
a Markovian time-homogeneous process -- a subclass of invertible Markov processes, which are described by a 
time-independent Liouvillian $\C L$ -- is classical with respect to measurements of a non-degenerate observable 
\emph{for an initially diagonal state} if and only if the dynamics are NCGD. 

Without the three assumptions of invertibility, Markovianity and non-degeneracy of the measured observable, notable 
differences start to appear. First, our definition of incoherent dynamics remains meaningful even if the dynamics 
are not invertible or if the measured observable is degenerate: in the first case, the dynamical map 
$\Lambda_{\ell,k}$ is not unambiguously defined for $t_k>t_0$ and in the second case, even $\Lambda_{k,0}$ might 
not be defined if the system remains entangled with the environment after an initial dephasing operation. Most 
notably, however, in the non-Markovian regime Eq.~(\ref{eq NCGD}) cannot directly be checked in an 
experiment by comparing two sets of ensembles, one which was dephased in the middle of the evolution and one which 
was not. Indeed, if the dynamics are non-Markovian, then the dynamics after a dephasing operation at time $t_k$ 
will not be described by the map $\Lambda_{\ell,k} = \Lambda_{\ell,0}\Lambda_{k,0}^{-1}$. 
We will exemplify this point by an example, which was also considered in Refs.~\cite{PollockEtAlPRL2018, 
SmirneEtAlQST2018} and experimentally realized in Ref.~\cite{LiuEtAlNatPhys2011}. 

The model describes a spin coupled to 
a continuous degree of freedom via the Hamiltonian $H_{SE} = \frac{g}{2} \sigma_z\otimes\hat q$. The initial state of 
the environment is taken to be pure with a wavefunction in coordinate representation 
$\psi_E(q) = \sqrt{\gamma/\pi}/(q+i\gamma)$. For an initially decorrelated system-environment state the exact 
reduced system dynamics are 
$\rho(t) = \mbox{tr}_E\{e^{-iH_{SE}t}\rho(0)\otimes|\psi\rangle_E\langle\psi|e^{-iH_{SE}t}\}$. Evaluating the trace 
in the coordinate basis and using $e^{i\alpha\sigma_z} = \cos\alpha + i\sin\alpha \sigma_z$, 
it is easy to confirm that 
\begin{equation}
 \begin{split}
  \rho(t) = \int dq  
  & \frac{\gamma/\pi}{q^2+\gamma^2}\left[\cos\left(\frac{gq}{2}t\right) - i\sin\left(\frac{gq}{2}t\right)\sigma_z\right] \\
  & \times \rho(0)\left[\cos\left(\frac{gq}{2}t\right) + i\sin\left(\frac{gq}{2}t\right)\sigma_z\right].
 \end{split}
\end{equation}
Explicit evaluation of the integrals yields 
\begin{equation}\label{eq example 1}
 \rho(t) = \frac{1}{2}(1+e^{-\Gamma t})\rho(0) + \frac{1}{2}(1-e^{-\Gamma t})\sigma_z\rho(0)\sigma_z,
\end{equation}
where we have introduced the dephasing rate $\Gamma \equiv g\gamma$. Next, we take Eq.~(\ref{eq example 1}), substract 
$\sigma_z$(\ref{eq example 1})$\sigma_z$ and multiply by $e^{-\Gamma t}/2$ to confirm that 
\begin{equation}
 \frac{e^{-\Gamma t}}{2}[\rho(0) - \sigma_z\rho(0)\sigma_z] = \frac{1}{2}\rho(t) - \frac{1}{2}\sigma_z\rho(t)\sigma_z.
\end{equation}
This allows us to deduce a master equation for the two-level system by taking the time-derivative of 
Eq.~(\ref{eq example 1}) and by using the previous result: 
\begin{equation}
 \begin{split}
  \partial_t\rho(t) &=  -\Gamma \frac{e^{-\Gamma t}}{2}[\rho(0) - \sigma_z\rho(0)\sigma_z]  \\
                    &=  \frac{\Gamma}{2}[\sigma_z\rho(t)\sigma_z-\rho(t)]  \\
                    &=  \frac{\Gamma}{2}\left[\sigma_z\rho(t)\sigma_z-\frac{1}{2}\{\sigma_z^2,\rho(t)\}\right] 
                    =\equiv\C L\rho(t),
 \end{split}
\end{equation}
where $\C L$ denotes the `Liouvillian'. This is a very simple master equation where the expectation values of the Pauli 
matrices $[x(t),y(t),z(t)] = [\lr{\sigma_x}(t),\lr{\sigma_y}(t),\lr{\sigma_z}(t)]$ obey the differential equations 
\begin{equation}
 \dot x(t) = -\Gamma x(t), ~~ \dot y(t) = -\Gamma y(t), ~~ \dot z(t) = 0.
\end{equation}
The solution of these equations is obvious. 

Next, let us apply a dephasing operation in the $\sigma_x$ basis at an arbitrary time $s>0$, which is defined for any 
$\rho$ as 
\begin{equation}
 \Delta\rho = |+\rl+|\rho|+\rl+| + |-\rl-|\rho|-\rl-|,
\end{equation}
where $|\pm\rangle = (|1\rangle\pm|0\rangle)/\sqrt{2}$. 
Note that for a density matrix parametrized by a Bloch vector $(x,y,z)$ in the $\sigma_z$ basis we obtain 
\begin{equation}\label{eq dephasing op qubit}
 \Delta\frac{1}{2}\begin{pmatrix}
        1+z     &   x-iy    \\
        x+iy    &   1-z \\
       \end{pmatrix} = \frac{1}{2}
       \begin{pmatrix}
        1   &   x   \\
        x   &   1   \\
       \end{pmatrix}.
\end{equation}
We now want to compute the exact system state at time $t>s$ after a dephasing operation was applied, i.e., 
\begin{equation}
 \rho(t) = \mbox{tr}\{\C U_{t,t-s}\Delta\C U_{s,0}\rho(0)\otimes|\psi\rangle_E\langle\psi|\}.
\end{equation}
By evaluating the trace again in the coordinate representation, this can be done straightforwardly although the 
calculation becomes now more tedious. The result for an initial state 
with expectation value $\lr{\sigma_x}(0) = x_0$ [the other expectation values do not matter as they get erased in 
the dephasing operation, cf.~Eq.~(\ref{eq dephasing op qubit})] is 
\begin{equation}
 \begin{split}\label{eq exact x}
  x(t)  =&~  \frac{x_0}{2} \left\{\cosh[\Gamma(t-2s)] + \cosh(\Gamma t)\right\}  \\
        &-  \frac{x_0}{2} \left\{\sinh(\Gamma t) + \frac{\sinh[\Gamma(t-2s)]}{\text{sign}(t-2s)}\right\}.  
 \end{split}
\end{equation}

\begin{figure}[h]
 \centering\includegraphics[width=0.35\textwidth,clip=true]{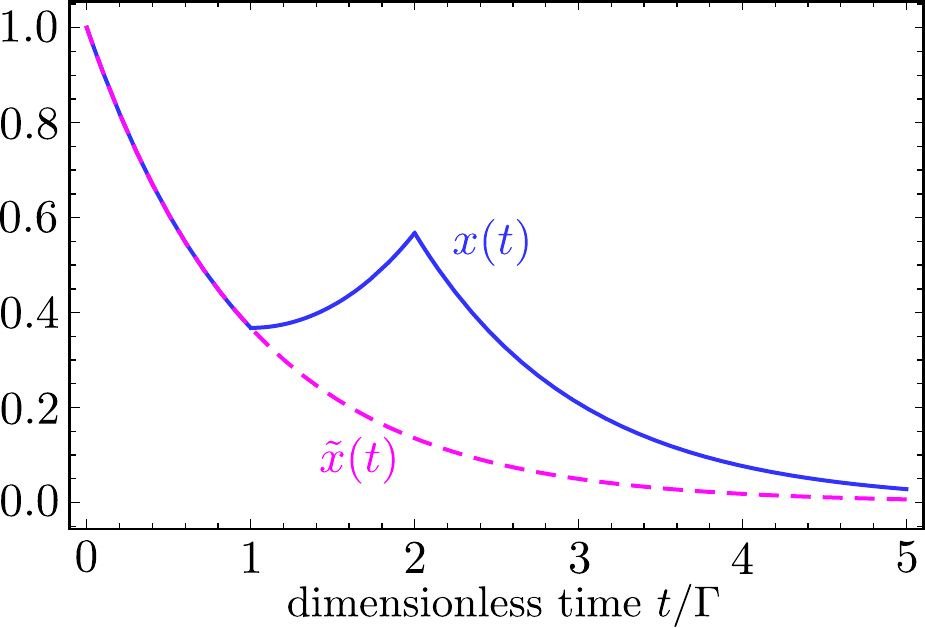}
 \label{fig plot example} 
 \caption{Plot of the exact time evolution [Eq.~(\ref{eq exact x}), solid blue line] compared with the approximated 
 one [Eq.~(\ref{eq approx}), dashed pink line]. Parameters are $\Gamma = 1, s = 1$ and $x_0 = 1$. }
\end{figure}

Now, for time-homogeneous dynamics the definition of NCGD in Ref.~\cite{SmirneEtAlQST2018} reduces to 
\begin{equation}\label{eq NCGD def Smirne}
 \Delta e^{\C L(t-s)}\Delta e^{\C Ls}\Delta = \Delta e^{\C L t}\Delta
\end{equation}
for all $t>s>0$. For our example we get according to the dynamics in Eq.~(\ref{eq NCGD def Smirne}) 
\begin{equation}\label{eq approx}
 \tilde x(t) = \mbox{tr}\{\sigma_x e^{\C L(t-s)}\Delta e^{\C Ls}\Delta \rho(0)\} = e^{-\Gamma t} x_0
\end{equation}
for all $t$ and especially independent of any dephasing operation. Hence, the dynamics is NCGD according to 
the definition from Ref.~\cite{SmirneEtAlQST2018}. But by looking at the exact time-evolution of the system 
[cf.~Eq.~(\ref{eq exact x}) and Fig.~\ref{fig plot example}], we see that even the mean value $x(t)$ can show a 
strong dependence on the dephasing operation. Therefore, according to our definition, the dynamics are 
\emph{not} incoherent with respect to the $\sigma_x$ basis. 

Finally, we mention that there are a couple of finer details too. For instance, in our work we only consider a fixed 
set of discrete times whereas Smirne \emph{et al.}~allow for arbitrary times. On the other hand, the system observable 
$R_k$ was not allowed to be explicitly time-dependent in Ref.~\cite{SmirneEtAlQST2018}. These points can be, however, 
incorporated in each of the frameworks and therefore we did not put any additional emphasis on these minor details. 

\section{Counterexamples}
\label{sec app counterexamples}

\subsubsection*{A process which is incoherent for one preparation $\C A_0$ but not classical for that preparation}

Consider an isolated two-level system undergoing purely unitary dynamics. Then, the dynamics are incoherent with respect 
to any preparation $\C A_0$ which maps the system state to a completely mixed state: independent of any dephasing or 
identity operation, the state will stay at the origin of the Bloch ball for all times. 
 
However, such a dynamics does not necessarily imply classical statistics. Consider, e.g., the measurement basis to be 
$\sigma_z$ (with outcomes $\{\uparrow_k,\downarrow_k\}$ at times $t_k$) and the unitary rotations to be around the 
$y$-axis. Furthermore, the time-steps are chosen equidistant in such a way that the rotation is exactly $\pi/2$. 
It is then easy to confirm that 
\begin{equation}
 p(\uparrow_3,\uparrow_2,\uparrow_1) = p(\uparrow_3,\downarrow_2,\uparrow_1) = \frac{1}{8},
\end{equation}
hence, $\sum_{\sigma_2\in\{\uparrow,\downarrow\}} p(\uparrow_3,\sigma_2,\uparrow_1) = 1/4$. But if we do not 
perform any measurement at time $t_2$, we obtain ${p(\uparrow_3,\cancel{\sigma_2},\uparrow_1) = 0}$. 
The statistics are therefore non-classical. 

\subsubsection*{A process which is Markovian and incoherent for all preparations but not classical}

Consider a Markov process for a two-level system where the map in the first time-step is defined by 
\begin{equation}
 \Lambda_{2,1}: \rho\mapsto\frac{1}{2}\begin{pmatrix}
                                        1 & 0 \\ 0 & 1 \\
                                      \end{pmatrix}
\end{equation}
for any input state $\rho$. The rest of the dynamics is again unitary as in the previous counterexample. 
Thus, the dynamics are incoherent for any preparation, but not classical. 

\subsubsection*{A process which is Markovian, invertible and classical for all preparations but not incoherent with 
respect to measurements of a degenerate observable}

Consider two qubits $A$ and $B$ and projective measurements in some fixed basis of qubit $A$ only such that the 
dephasing operation acts only locally on qubit $A$: $\Delta = \Delta_A\otimes\C I_B$. Thus, the measured observable 
is degenerate and projects onto two possible subspaces of dimension two. Furthermore, we only consider measurements at 
two times $t_2$ and $t_1$ and assume the dynamics in between these two times to be described by a unitary swap gate, 
$U_\text{swap}|i_Aj_B\rangle = |j_Bi_A\rangle$. We also assume that the dynamics in between the preparation 
and the first measurement is trivial, i.e., described by an identity operation. 

Now, consider an arbitrary initial state resulting from an arbitrary preparation $\C A_0$, denoted as 
\begin{equation}
 \C A_0\rho_0 = \sum_{i_A,i_B,j_A,j_B} \rho_{i_Ai_B,j_Aj_B} |i_Ai_B\rl j_Aj_B|.
\end{equation}
Then, straightforward calculation reveals that 
\begin{align}
 p(r_2,r_1)   &=  \mbox{tr}_{AB}\{\C P_{r_2}\C U_\text{swap}\C P_{r_1}\C A_0\rho_0\} = \rho_{r_1r_2,r_1r_2}, \\
 p(r_2,\cancel{r_1})   &=  \mbox{tr}_{AB}\{\C P_{r_2}\C U_\text{swap}\C A_0\rho_0\} = \sum_j\rho_{jr_2,jr_2}.              
\end{align}
Hence, the process is classical: $\sum_{r_1} p(r_2,r_1) = p(r_2,\cancel{r_1})$. 
 
However, the process is not incoherent. Consider, for instance, the initial state 
\begin{equation}
 \C A_0\rho_0 = |\psi_0\rl\psi_0|, ~~~ |\psi_0\rangle = \frac{|0_A\rangle+|1_A\rangle}{\sqrt{2}}\otimes|0_B\rangle.
\end{equation}
Then, 
\begin{equation}
 \mf T_3[\Delta_2,\C I_1,\C A_0] = |\psi_1\rl\psi_1|, ~~~ |\psi_1\rangle = |0_A\rangle\otimes\frac{|0_B\rangle+|1_B\rangle}{\sqrt{2}},
\end{equation}
but 
\begin{equation}
 \mf T_3[\Delta_2,\Delta_1,\C A_0] = \frac{|0_A0_B\rl0_A0_B| + |0_A1_B\rl0_A1_B|}{2}.
\end{equation}

\end{document}